%% file: 0utilcache2.0.tex
\documentclass[10pt, conference, letterpaper]{IEEEtran}
\pdfoutput = 1
\IEEEoverridecommandlockouts
\usepackage{cite}
\usepackage{amsmath}
\usepackage{booktabs}
\usepackage{amssymb,amsfonts}
\usepackage{mathrsfs}
\usepackage{graphicx}
\usepackage{textcomp}
\usepackage{xcolor}
\def\BibTeX{{\rm B\kern-.05em{\sc i\kern-.025em b}\kern-.08em
    T\kern-.1667em\lower.7ex\hbox{E}\kern-.125emX}}
    
\usepackage{amsthm}
\newtheorem{theorem}{Theorem}

\newtheorem{propo}{Proposition}

\usepackage{algorithm}               
\usepackage{algorithmic}             
\usepackage{multirow}                
\usepackage{graphics}

\usepackage{url} 
\makeatletter
\def\url@leostyle{%
  \@ifundefined{selectfont}{\def\UrlFont{\sf}}{\def\UrlFont{\small\ttfamily}}}
\makeatother
\usepackage{subfigure} 

\usepackage{threeparttable} 

\begin{document}

\title{UtilCache: Effectively and Practicably Reducing Link Cost in Information-Centric Network}

\author{\IEEEauthorblockN{Lemei Huang\IEEEauthorrefmark{1},
Yu Guan\IEEEauthorrefmark{1}\IEEEauthorrefmark{2},
Xinggong Zhang\IEEEauthorrefmark{1}\IEEEauthorrefmark{2}\IEEEauthorrefmark{3} and
Zongming Guo\IEEEauthorrefmark{1}\IEEEauthorrefmark{2}\IEEEauthorrefmark{3}}
\IEEEauthorblockA{\IEEEauthorrefmark{1}Institute of Computer Science Technology, Peking University, China, 100871}
\IEEEauthorblockA{\IEEEauthorrefmark{2}PKU-UCLA Joint Research Institute in Science and Engineering}
\IEEEauthorblockA{\IEEEauthorrefmark{3}Cooperative Medianet Innovation Center, Shanghai, China}
\IEEEauthorblockA{\{milliele, shanxigy, zhangxg, guozongming\}@pku.edu.cn}}

\maketitle

\begin{abstract}
Minimizing total link cost in Information-Centric Network (ICN) by optimizing content placement is challenging in both effectiveness and practicality. To attain better performance, upstream link cost caused by a cache miss should be considered in addition to content popularity. To make it more practicable, a content placement strategy is supposed to be distributed, adaptive, with low coordination overhead as well as low computational complexity. In this paper, we present such a content placement strategy, UtilCache, that is both effective and practicable. UtilCache is compatible with any cache replacement policy. When the cache replacement policy tends to maintain popular contents, UtilCache attains low link cost. In terms of practicality, UtilCache introduces little coordination overhead because of piggybacked collaborative messages, and its computational complexity depends mainly on content replacement policy, which means it can be $O(1)$ when working with LRU. Evaluations prove the effectiveness of UtilCache, as it saves nearly 40\% link cost more than current ICN design.
\end{abstract}


\input{1introduction}

\input{2related-work}

\input{3CU}

\input{4WeightedCache}

\input{5eva}

\section{Conclusion and Future Work}
\label{conclusion}
Formulation of Link Cost Minimization problem in ICN reveals the importance of upstream link cost. In this paper, we present UtilCache, an adaptive, distributed content placement strategy to reduce total link cost, which is both effective and practicable. In UtilCache, each router tends to cache contents with highest caching utility. Caching utility is defined considering both content popularity and upstream link cost, which ensures the performance of UtilCache. Moreover, UtilCache introduces little coordination overhead because of piggybacked collaborative messages, and is compatible with any cache replacement policy, which means its computational complexity can be low when efficient content replacement policy is used. Evaluations validates the effectiveness of UtilCache as well. 
UtilCache can be further improved in terms of coordination overhead by maintaining only the information of in-cache contents, such as In-CacheLFU \cite{incachelfu}. In addition to these avenues, we are going to perform an extended evaluation of UtilCache in a scenario with dynamic demand. 




\bibliographystyle{IEEEtran}
\bibliography{sigproc}
%
\include{7appendix}


\end{document}

%% file: 1introduction.tex
\section{Introduction}
Information-Centric Network (ICN) is a future Internet architecture proposed to achieve efficient content retrieval and distribution\cite{icn}. In ICN, caches are ubiquitous within routers, a.k.a.\ in-network caching. 
Whenever a packet traverses a link, there arises cost (propagation latency, money, bandwidth occupation etc.). Once a content request is satisfied at an intermediate cache on its path to content source, upstream links will never be traversed during this session, and thus the potential link cost are saved (Fig.~\ref{lcmpic}).
It is possible to minimize total link cost by optimizing content placement, which we refer to as \emph{Link Cost Minimization} (LCM) problem .


Current ICN design adopts a content placement strategy named Leave Copy Everywhere (LCE): contents are replicated and cached in every intermediate router on the response path to requester (e.g.\ from "I" to "R" in Fig.~\ref{lcmpic}). Cache replacement is done individually at each router with simple policy such as LRU or LFU, which tends to cache locally popular contents. Simple as it is, LCE works unsatisfactorily in link cost reducing \cite{26}.

\begin{figure}[thb]
\centering
\includegraphics[width=0.35\textwidth]{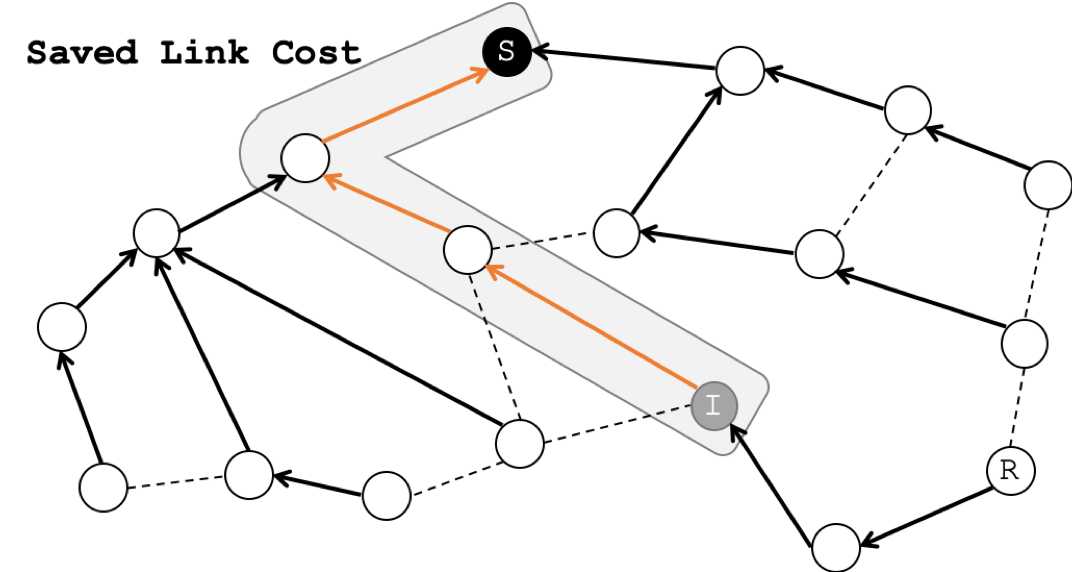}
\caption{\textbf{Link Cost Minimization in ICN:} Caches are ubiquitous in ICN. Requests for a specific content are forwarded to its designated source ("S"). If a request from R is satisfied at I, the potential link cost (from "I" to "S") is saved.
}
\label{lcmpic}
\end{figure}

We believe that designing a content placement strategy to reduce link cost faces two challenges:

\textbf{Effectiveness: }Most significantly, the strategy is supposed to perform well in reducing link cost. Many works \cite{magic,26}, including our previous work \cite{utilcache}, have noticed that in order to cut down link cost, in addition to popularity, upstream link cost caused by a cache miss should also be taken into account.

\textbf{Practicality: }First, the strategy must be \emph{distributed}, because a centralized controller is very likely to overburden due to network scale and content amount. Second, since user demand is dynamic and not a priori known, the strategy should be \emph{adaptive}. Third, there may be extra communication and storage overhead (collectively called \emph{coordination overhead}) introduced in order to attain higher performance, but overmuch coordination overhead could possibly in turn damage the performance. Last but not the least, \emph{computational complexity} of the strategy should be noticed as well. Computationally expensive algorithm is not applicable in large network, for which LRU is more commonly used than LFU.

Typically, LCE is a strategy well-practicable but ineffective. 
In this paper, we present \emph{UtilCache}, a distributed and adaptive content placement strategy in ICN, which yields satisfactory performance in reducing link cost, introduces little coordination overhead, and can be computationally efficient.
Our main contribution of this paper are summarized as:
\begin{itemize}
\item We present \emph{caching utility}, which quantifies the caching gain (i.e. saved link cost) of caching a content at a node (Sec.~\ref{globalmodel}). Caching contents by caching utility ensures the effectiveness of this strategy. Similar things are done in our previous work \cite{utilcache} roughly and intuitively. However, in this paper, we elaborate its derivation. Theoretical analysis begins with LCM problem formulation, and then we give a 1/2-approximation offline algorithm to solve it, from which we derive caching utility.
\item We propose a distributed and adaptive content placement strategy, UtilCache, which can work with any cache replacement policy to keep the contents with high caching utility in the cache (Sec.~\ref{algorithm}). UtilCache introduces little storage overhead compared with LCE, and yields the same computational complexity as LCE when they adopts same content replacement policy.
\end{itemize}

%% file: 2related-work.tex
\section{Related Work}
\label{relatedwork}
The goal of optimizing content placement in ICN mainly focus on: 1) reducing cache redundancy and keeping content diversity, and 2) reducing total link cost. ProbCache\cite{probcache}, together with \cite{intraas,lessformore}, are all instances of cache redundancy reduction. In this paper, we focus on other works coping with link cost reduction.

Simple and practicable as it is, LCE has great weakness in its performance. \cite{26} has proven LCE's performance can be arbitrarily suboptimal when with simple cache replacement policies such as LRU, LFU. Both LFU and LRU in fact concerns only the local popularity of contents. As shown in Fig.~\ref{weakness}, LCE neglects the upstream link cost that may be paid due to a cache miss.

To improve effectiveness, proposed content placement strategies in regard to link cost reduction claim to focus on caching utility, which refers to the saved link cost on account of a cached content. To our knowledge, \cite{lcm1} is the first work concerning caching utility, yet without elaboration of it. In 2014, Ren et al.\ present MAGIC\cite{magic}, empirically and intuitively giving the definition of caching utility. Our previous work \cite{utilcache} roughly derives caching utility from a LCM problem formulation. Ioannidis et al.\ \cite{26} use Pipage Rounding \cite{pipage} method to solve LCM problem, and present an adaptive and distributed content placement strategy, PGA, which performs with optimality guarantees.

However, practicability becomes a drawback for these strategies. For example, content popularity should be a priori known in \cite{lcm1,utilcache}. Extra messages are disseminated in PGA, which increase traffic cost. To address the practicability issues, another strategy, GRD, is proposed in \cite{26}, where routers decide content placement based on the estimation of a sub-gradient and always keep the higher ones. The "sub-gradient" is similar to caching utility we refer to, but not exactly identical. The estimation of such sub-gradient is to maintain an exponentially weighted moving average (EWMA) of the upstream link cost that is piggybacked by data packets. GRD keeps adaptive and distributed yet cuts down coordination overhead. However, it faces two problems. First, to maintain EWMA, each time a data packet arrives, the estimation of all contents should be modified which is computationally expensive. Second, locally cached popular contents may be regarded as unpopular ones because requests for cached contents are locally satisfied, and thus there is not data packet from upstream, which indicates that EWMA of upstream link cost may not a satisfactory metric to measure caching utility.

\begin{figure}[tb]
\centering
\includegraphics[width=0.4\textwidth]{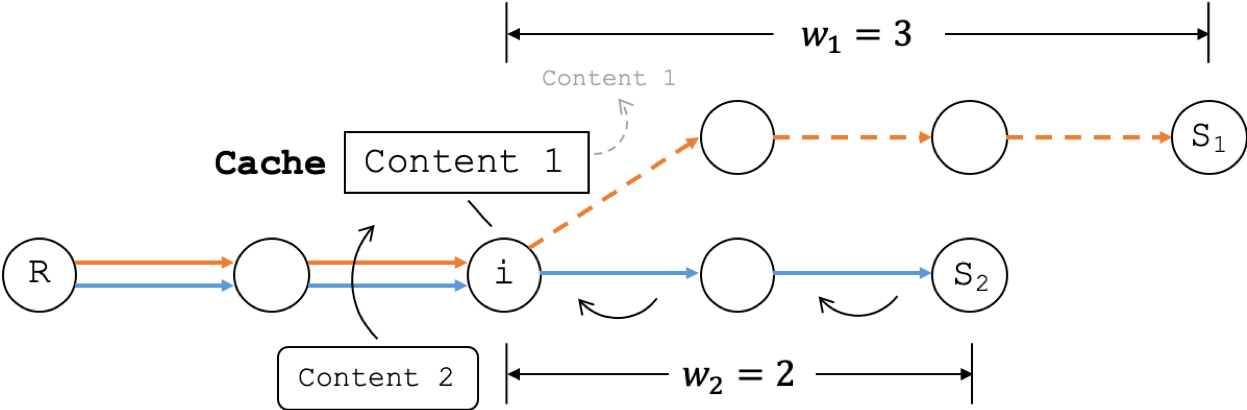}
\caption{\textbf{Weakness of LCE: } Suppose content 1 and 2 has the same popularity. Ideally it is better to cache content 1 at router $i$ because a cache miss for content 1 leads to more upstream link cost. In LCE, however, only popularity is considered, for which content 1 is replaced by content 2. }
\label{weakness}
\end{figure}

%% file: 3CU.tex
\section{Background}
The derivation of caching utility begins with a formulation of LCM problem in ICN, which we present in this section.

\subsection{Assumptions}



\textbf{Independent Reference Model. } Contents fall into a fixed catalogue of $|\mathbf{C}|$ objects, and requests for each content are generated with constant probability, independent from any past requests, a.k.a.\ Independent Reference Model (IRM) \cite{IRM1}. We refer to popularity as \emph{per-content request rate} hereafter.
IRM ignores the \emph{temporal locality}, a significant feature in real content request traces. Although a Shot Noise Model (SNM) is presented in \cite{temporal} to capture the dynamics of content popularity, Analytical models under SNM are shown to be challenging \cite{snmdifficult}. Therefore, we still consider the IRM in our model.


\textbf{Equal-sized Contents. } Every content in the network has unified size. Since the derivation is similar without this assumption, we still keep it for simplicity considerations.

\textbf{No Request Aggregation. } We suppose there exists no content request aggregation mechanism such as the well-known PIT in Name-Data Network (NDN)\cite{ndn}, as \cite{ccnramp} have expounded PIT does not deliver substantial benefits but is likely to cause problems counter-productively.

\textbf{Symmetric Path. } Content request and the corresponding data packet traverse the same path, which is a common assumption in the context of ICN, similarly to \cite{probcache,intraas,26}.

\textbf{Instantaneous Data Response. } Data packets are downloaded instantaneously, or at least in a small delay negligible compared with the request arrival process, which results in the equal arrival rate for both requests and contents.

\subsection{Modeling Link Cost Minimization Problem}

The whole cache network can be represented as an graph $\mathbf{G}=\langle\mathbf{N}, \mathbf{E}\rangle$, where $\mathbf{N}=\{1,...,n\} $ is a set of nodes as the routers and each edge $(i,j)$ in the edge set $\mathbf{E}$ indicates packets can be transmitted bidirectionally between node $i$ and $j$. Link cost, the most significant element in LCM problem, can thus be formulated as the weight $w_{ij}$ of each edge $(i,j)$.

It can be generally assumed that each node $i\in \mathbf{N}$ is equipped with a cache able to store any content $c\in\mathbf{C}$. Thereafter, we use "node", "router" and "cache" interchangeably. The cache capacity of node $i$ is $B_i$. Specifically, a set of nodes $\mathbf{S_c}$ are source nodes for content $c$, which stores $c$ permanently.

Cache decisions, \emph{the only independent variables} of our formulation, are denoted by a matrix $\mathbf{X}=[x^c_i]$ with $x^c_i\in\{0,1\}$. $x^c_i$ is binary indicating whether content $c$ is cached in node $i$ ($x^c_i=1$ if being cached).

Ioannidis et al.\ \cite{26} presents an elegant way to formulate the routing strategy in ICN. Let $p=\{p_1, p_2, ...,p_K\}$ be a simple (i.e. without loops) path end with a source node, none of the intermediate nodes in $p$ source node. Thus, all the possible requests can be defined as a set $\mathcal{R}$. Each $=(c,p)\in\mathcal{R}$ refers to the requests for $c$ generated from $p_1$, forwarded along $p$ and satisfied at $p_K$, which indicates that $p_K$ is one of the source nodes of content $c$, i.e. $p_K\in\mathbf{S_c}$.

Let $\lambda_{(c,p)}$ be arrival rate of requests falling into $(c,p)$. Total link cost of network $L(\mathbf{X})$ can thus be defined as:
\begin{equation}
\label{lc-raw}
L(\mathbf{X})=\sum_{(c,p)\in\mathcal{R}}\lambda_{(c,p)}\sum_{k=1}^{K-1}w_{p_{k+1}p_k}\prod_{k'=1}^k (1-x^c_{p_{k'}})
\end{equation}

\eqref{lc-raw} shows that a request for content $c$ on path $p$ contributes $w_{p_{k+1}p_k}$ to the total link cost when none of $p_k$'s downstream nodes cache the requested content. For simplicity, we only calculate link cost caused by data packets. Now we can formulate LCM problem as:

\begin{subequations}
\begin{alignat}{2}
\min\quad & L(\mathbf{X})&{}& \label{GM:1}\\
\mbox{s.t.}\quad &\sum_{c\in\mathbf{C}}{x^c_i}\le B_i&{}&,\forall i\in\mathbf{N}\label{GM:2}\\
&x_i^c\in\{0,1\}&{}&,\forall c\in\mathbf{C},\forall i\in\mathbf{N}\label{GM:3}
\end{alignat}
\end{subequations}

\eqref{GM:2} and \eqref{GM:3} are capacity and integrality constraints, respectively.

When all the caches are empty, then total link cost is constant, independent from $\mathbf{X}$. Let $L_0$ be the constant cost. We have $L_0 = \sum_{(c,p)\in\mathcal{R}}\lambda_{(c,p)}\sum_{k=1}^{K-1}w_{p_{k+1}p_k}$.

The benefit of caching is to save total link cost. We refer to the saved link cost as \emph{caching gain}, similarly in \cite{magic, 26}, which is defined as $G(\mathbf{X})=L_0-L(\mathbf{X})$. Then we have:

\begin{equation}
\label{slc}
G(\mathbf{X})=\sum_{(c,p)\in\mathcal{R}}\lambda_{(c,p)}\sum_{k=1}^{K-1}w_{p_{k+1}p_k}(1-\prod_{k'=1}^k (1-x^c_{p_{k'}}))
\end{equation}

To minimize total link cost is to maximize caching gain. The LCM problem can be equivalently formulated as:

\begin{subequations}
\label{U}
\begin{alignat}{2}
\max\quad & G(\mathbf{X})\label{U:1}\\
\mbox{s.t.}\quad &\mathbf{X}\in\mathcal{D}\label{U:2}
\end{alignat}
\end{subequations}

$\mathcal{D}$ is the set of all the $\mathbf{X}$ satisfying \eqref{GM:2}-\eqref{GM:3}. Shanmugam et al.\ has proven \eqref{U} is NP-hard \cite{femto2013}, for which an approximation algorithm is needed.

\section{Caching Utility}
\label{globalmodel}
In this section, we present the our definition of \emph{caching utility}, which quantifies the caching gain (i.e.\ saved link cost in LCM problem) of caching a content in a router. Caching utility is derived from an offline 1/2-approximation algorithm solving \eqref{U}.

\subsection{Offline Algorithm}

Before presenting the offline algorithm, we transform \eqref{U} into maximizing monotone submodular function subject to matroid constraints:

\textbf{Properties of \eqref{U} (abstract). }The integrality constraint \eqref{GM:3} enables that every cache decision $\mathbf{X}=\{x^c_i\}$ can be written as a set $A\subset \{f^c_i|c\in\mathbf{C}, i\in\mathbf{N}\}$, where
$x^c_i=1 \Leftrightarrow f^c_i\in A$.
Thus, $\mathcal{D}$ can be written as matroid constraints, according to the definition of partition matroids\cite{partionmatroid}. Moreover, $G(\mathbf{X})$ can be written as a set function\cite{setfunction} as well. It is proved by \cite{femto2013} that the set function is a monotone submodular function, which means \eqref{U} can be written as an optimization problem that maximizing monotone submodular function subject to matroid constraints. Due to space constraints, detailed proof is in Appendix~\ref{transform}.

\begin{algorithm}[tbh]         
    \caption{Offline Algorithm} 
    \label{centralizedalgo} 
    \begin{algorithmic}[1]                
    \REQUIRE ~~\\                          
        The network state;\\
    \ENSURE ~~\\                           
        The cache decisions,  $\mathbf{X}$;
    \STATE $\mathbf{X}=\{\mathbf{0}\}$; 
    
    \WHILE{there is $(c,i)$ s.t.\ $x^c_i=0$ and $\sum_{c'\in\mathbf{C}}x^{c'}_i<B_i$}
    \STATE $(c,i) = argmax_{x^c_i=0\text{ and }\sum_{c'\in\mathbf{C}}x^{c'}_i<B_i}m^c_i(\mathbf{X})$;
    \STATE $\mathbf{X} = \mathbf{X}|x^c_i=1$; \label{change}
    \ENDWHILE

    \RETURN $\{\mathbf{X}\}$;                
    \end{algorithmic}
\end{algorithm} 

Fisher et al.\ \cite{greedy} presents a simple greedy algorithm to solve the optimization problem that maximizing monotone submodular function subject to matroid constraints, with specific optimality guarantees, based on which we present an offline algorithm. Before introducing the algorithm, we define the marginal value of caching content $c$ at node $i$ as 
$$
m^c_i(\mathbf{X}) = G(\mathbf{X}|x^c_i=1)-G(\mathbf{X}|x^c_i=0)
$$, where $\mathbf{X}|x^c_i=1$ refers to the new matrix generated by changing 
$x^c_i$ of $\mathbf{X}$ to 1. 

Then the offline algorithm is discribed in Algorithm~\ref{centralizedalgo}, which keeps on choosing greedily a pair of content $c$ and node $i$ with highest marginal value under capacity constraints \eqref{GM:2}, and then storing $c$ in $i$. \cite{greedy} has proven caching gain obtain by Algorithm~\ref{centralizedalgo} is at least 1/2 of the optimal. 


\subsection{Derivation of Caching Utility}
\label{caching}
We find that the marginal value $m^c_i(\mathbf{X})$ indicates the gain of caching a new content $c$ at $i$. Therefore, we can use $m^c_i(\mathbf{X})$ to quantify the aforementioned caching gain, which we called it \emph{caching utility}.

\begin{theorem}
\label{cu}
we have $$
m^c_i(\mathbf{X}) = \lambda_i^c(\mathbf{X})\bar{w}^c_i(\mathbf{X})
$$, where $\lambda_i^c$ refers to request arrival rate for content $c$ at $i$ (i.e.\ the popularity of $c$ at $i$), $\bar{w}^c_i$ refers to average upstream link cost per request generated when $i$ doesn't cache $c$.
\end{theorem}
Detailed proof of Theorem~\ref{cu} is in Appendix~\ref{cachingu}. 
Interestingly, both $\lambda_i^c(\mathbf{X})$ and $\bar{w}^c_i(\mathbf{X})$ are independent from the cache decisions at node $i$, which means they can be regarded as constant when we keep cache decisions of other nodes unchanged and adjust only the cache of node $i$, due to which we define the \emph{caching utility} for content $c$ at $i$ as
\begin{equation}
U(c,i) = \lambda^c_i\bar{w}_i^c
\end{equation}

%% file: 4WeightedCache.tex
\section{UtilCache}
\label{algorithm}

Running Algorithm~\ref{centralizedalgo} in centralized way, however, encounters practical obstacles. In this section, we propose an adaptive, distributed cache strategy, \emph{UtilCache}, to solve LCM problem, based on a simple idea of maximizing caching utility.



\subsection{Reshaping "Popularity"}
\label{intuitive}
The important idea implies in Algorithm~\ref{centralizedalgo} is that: \emph{to maximizing the caching gain $G(\mathbf{X})$, we should choose those contents with highest caching utility to cache}. In Algorithm~\ref{centralizedalgo}, $(c,i)$ is chosen globally. However, in a distributed algorithm, to avoid extra communication overhead, nodes are supposed to decide content placement within their own cache. 

An intuitive thought is that each node $i$ maintains the estimation of $\lambda_i^c$ and $\bar{w}^c_i$, prioritizes contents by their caching utility, and always caches the top-$B_i$ ones. However, on the one hand, it is computational expensive. On the other hand, in both our previous work \cite{utilcache} and recent papers \cite{popexp1,popexp2} we find popularity estimation is a tough task and error of popularity estimation makes great difference to cache performance.

Inspired by caching algorithm proposed in \cite{lrus} for Web Caching, we come up with \emph{UtilCache}. UtilCache tackles the two obstacles by separating procedures of content retrieval and cache update. \emph{Content retrieval} remains unchanged: when a request arrives at a node and its designated content is cached, the node responses with the content, otherwise the request is forwarded upstream to next hop.

In UtilCache, each node $i$ maintains an estimation of $\bar{w}^c_i$, which will be described in Sec.~\ref{w}. Let $\bar{w}_{max}$ be the maximal per-request upstream link cost among all the contents at $i$. Every request for $c$ at $i$ becomes an \emph{effective request} with probability $\bar{w}^c_i/\bar{w}_{max}$.

\begin{figure}[tbp]
\centering
\includegraphics[width=0.4\textwidth]{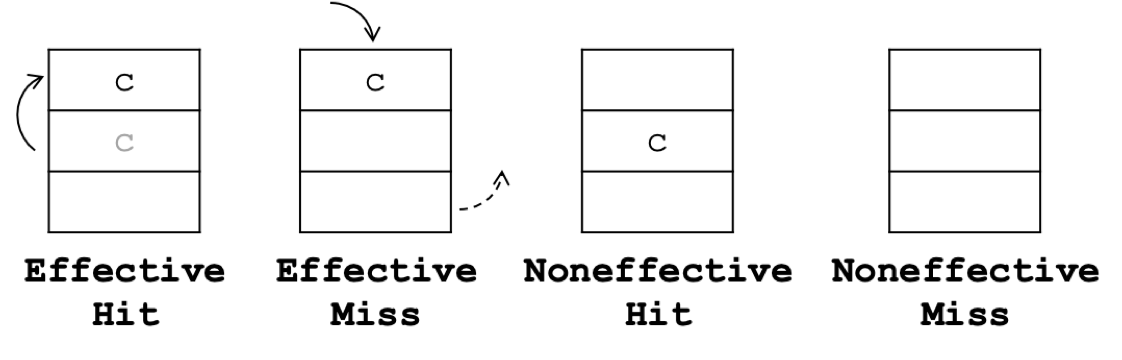}
\caption{\textbf{Cache Update: }Every request for $c$ at $i$ becomes \emph{effective request} with probability $\bar{w}^c_i/\bar{w}_{max}$. Cache update is different for effective and noneffective requests. Whenever an effective request arrives, cache is update accordingly, while noneffective request will not affect cache state. The example use LRU for content replacement. \label{update}}
\end{figure}

\emph{Cache update} differs between effective and noneffective requests.\emph{Effective requests} update cache state, while \emph{noneffective requests} do not, as if they have never arrived at the cache. Fig.~\ref{update} exhibits an example. Suppose LRU is used. When an effective request for $c$ results in a cache hit, $c$ becomes the most recently use content. Otherwise when a cache miss happens, $c$ is downloaded, cached and still the most recently use content. However, for a noneffective request, access time of the content remains unchanged. Even if the content is downloaded because of cache miss, it will not be cached.

Thus, cache-perceived popularity for content $c$ at node $i$ is equal to the popularity of effective requests, i.e.\ $$\lambda^c_i\frac{\bar{w}^c_i}{\bar{w}_{max}}=\frac{U(c,i)}{\bar{w}_{max}}$$. For a given node $i$, $\bar{w}_{max}$ is constant. Therefore in UtilCache, \emph{the most "popular" content in the cache's view, is the content with highest caching utility}.
UtilCache is compatible with any content replacement policy. Especially when working with those who tends to keep "popular" contents in cache (e.g.\ LFU, LRU), UtilCache can attain less link cost. 

\begin{table}[tb]
\centering
  \caption{Coordination Overhead and Complexity Comparision}
  \label{complex}
  \begin{threeparttable}
  \begin{tabular}{cccc}
    \toprule
      & \multirow{2}{*}{Computational Complexity}  &\multicolumn{2}{c}{Coordination Overhead} \\
      &   &  Space  &  Communication  \\
    \midrule
    GRD  &  $O(|\mathbf{C}|)$  &  $O(|\mathbf{C}|)$  &  0  \\
    U-LFU\tnote{1}  &  $O(log|\mathbf{C}|)$  &  $O(|\mathbf{C}|)$  &  0  \\
    U-LRU   &  $O(1)$  &  $O(|\mathbf{C}|)$ &  0  \\
        LCE-LFU  &  $O(log|\mathbf{C}|)$\tnote{2}&  0  &  0  \\
    LCE-LRU  &  $O(1)$  &  0  &  0    \\
    \bottomrule
  \end{tabular}
  \begin{tablenotes}
        \footnotesize
        \item[1] UtilCache with LFU.
        \item[2] In common implementation of LFU, $O(log|\mathbf{C}|)$ for eviction.
      \end{tablenotes}
  \end{threeparttable}
\end{table}

\subsection{Estimation of Per-request Upstream Link Cost}
\label{w}
In practice of UtilCache, each node is supposed to maintain the estimation of average upstream link cost $\bar{w}$. It occurs to us that the information of upstream link cost can be piggybacked by the returned data packets. We add a field for \emph{Upstream Link Cost} estimation in the data packet, named ULC. If a request from $s$ is satisfied at $t$, ULC of the data packet generated at $t$ is initialized to 0. Whenever it is transferred from $j$ to $i$ via link $(j,i)$, ULC in it is added by $w_{ji}$. Intermediate nodes can update their estimation by sniffing the ULC field of the arriving data packets. To emphasize the importance of fresh statistics, the average upstream link cost is calculated as a moving average:
$\bar{w}^c_i=\alpha w^c_i+(1-\alpha)\bar{w}^c_i$.

\subsection{Coordination Overhead and Computational Complexity}
Coordination overhead and computational complexity are summarized in Tab.~\ref{complex}, with some baselines.

Computational complexity of UtilCache depends mainly on the content replacement policy, because we introduce only $O(1)$ times of floating point computation.
Compared with LCE, UtilCache yields less link cost, and additionally introduces little coordination overhead: $O(|\mathbf{C}|)$ space cost for estimation of $\bar{w}$ and no extra packet for communication. With similar coordination overhead, UtilCache is more efficient than GRD, because GRD modifies the estimation of all contents each time a data packet arrives.
 
Moreover, UtilCache's compatibility with content replacement policy indicates that \emph{any progress made in content replacement policy to cache popular contents efficiently and effectively also does good to UtilCache.} For example, the $O(1)$ implementation of LFU \cite{o1} presented in 2010 can be applied in UtilCache and thus make the algorithm more computationally efficient while maintaining its effectiveness.

%% file: 5eva.tex
\begin{figure}[tb]
\centering
\subfigure[Caching gain when popularity skewness varies
\label{single-1}]{
\centering
\includegraphics[width=0.4\textwidth]{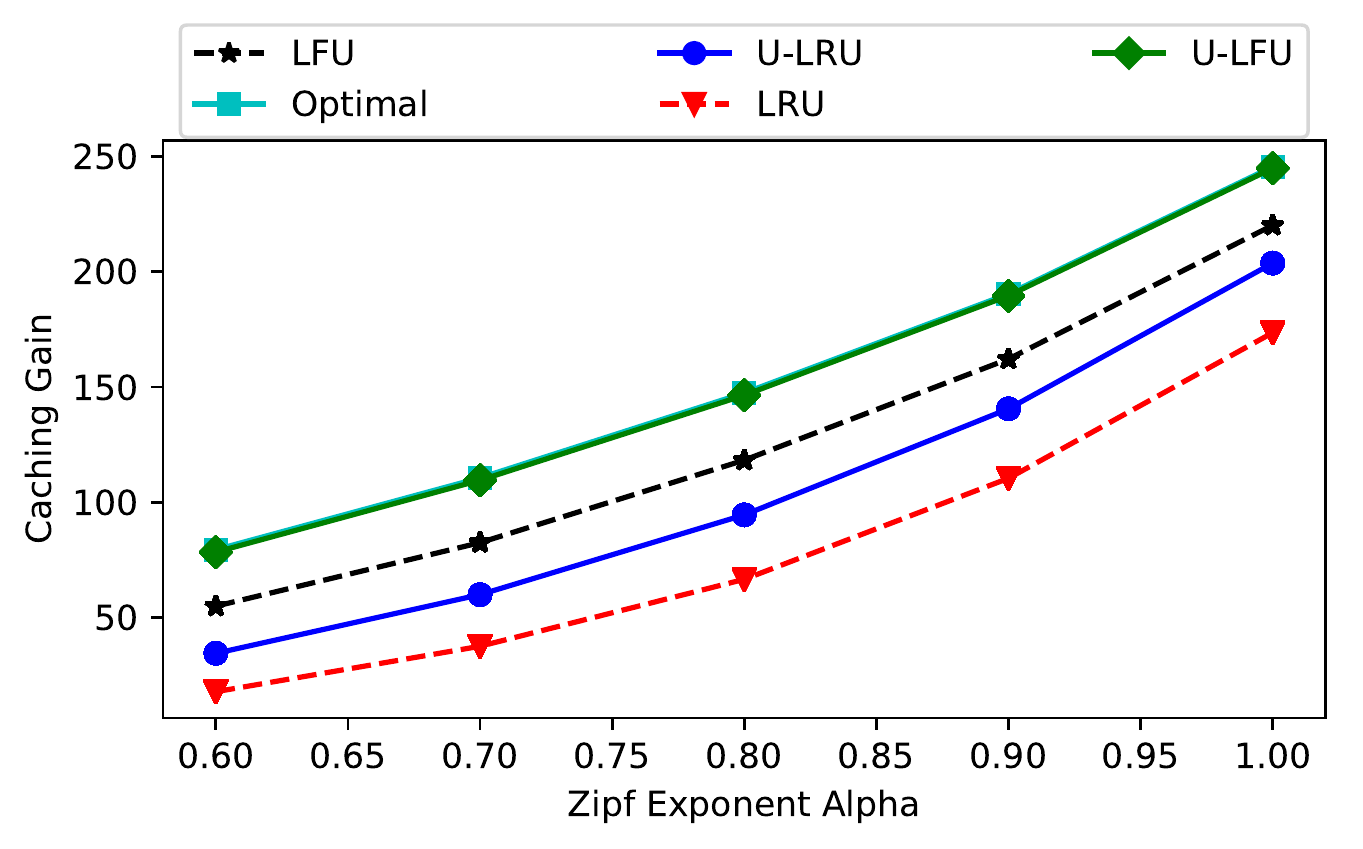}
}
\hspace{1pt}
\subfigure[Caching gain when cache size varies
\label{single-2}]{
\centering
\includegraphics[width=0.4\textwidth]{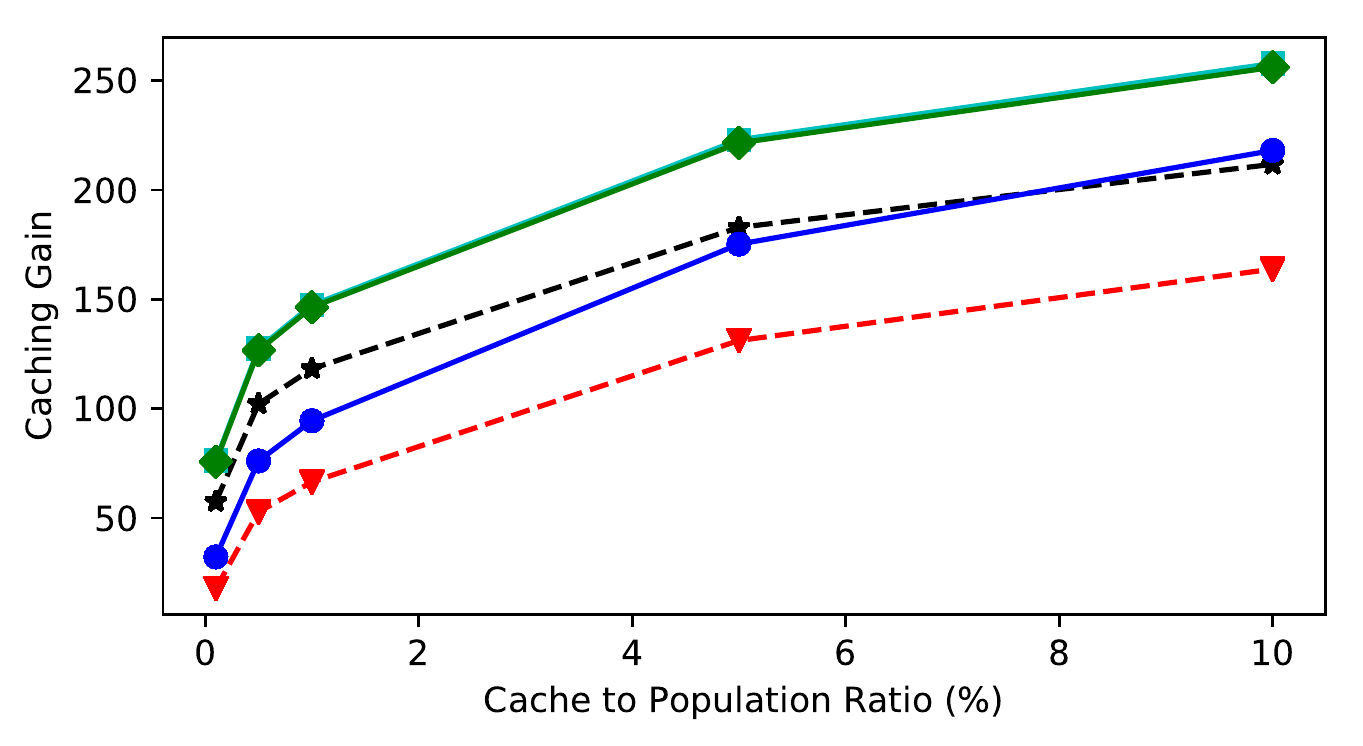}
}
\caption{Scenario 1: There is only one cache, and link cost added by a miss request for different content is different. When working with LFU, UtilCache performs nearly as well as the optimal solution (which is to always keep the contents with highest caching utility a prior).
\label{single}}
\end{figure}

\section{Performance Evaluation}
\label{eva}

\begin{figure*}[htbp]
\centering
\centering
\includegraphics[width=0.95\textwidth]{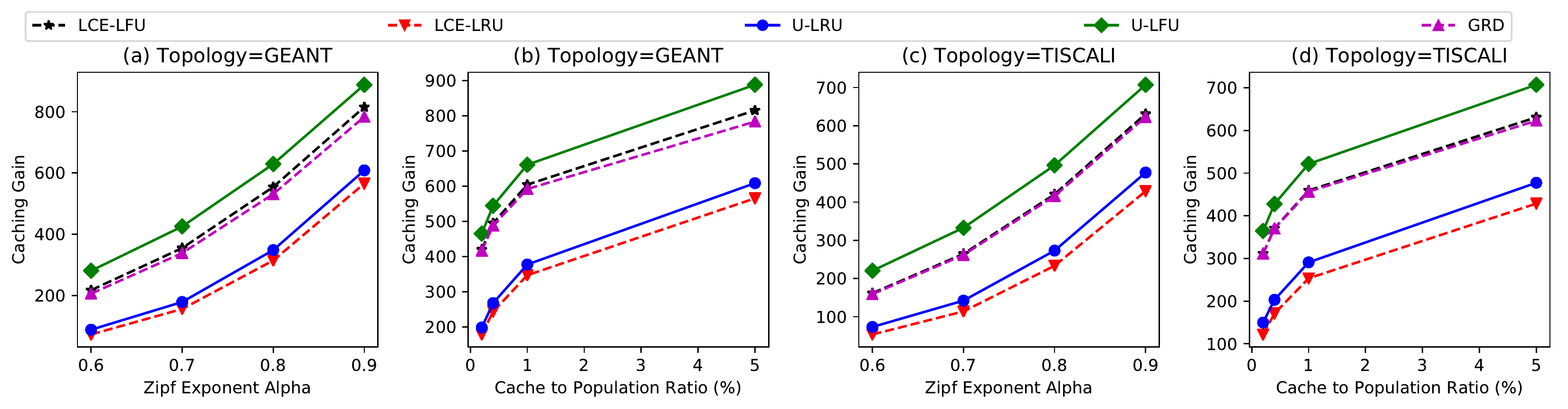}
\caption{Scenario 2: There is a cache network over specific topology, and cost of each link is assigned. UtilCache with LFU performs the best with most caching gain. }
\label{multi}
\end{figure*}

We conduct our evaluations on \emph{Icarus} \cite{icarus}, a discrete-event simulator offering flow-level simulations, to compare the performance of UtilCache, with benchmark cache strategies. All the evaluations are run multiple times.

\textbf{Networks. }We use two real network topologies: GEANT \cite{geant} and TISCALI\cite{tiscali}. GEANT is a core network interconnecting several European research institutes and universities, while TISCALI is from the RocketFuel dataset\cite{rock}. To show the performance difference better, we randomly assign each edge a \emph{link cost} choosing from $\{10, 100, 1000\}$.

\textbf{Workload. }There is a catalog of $3\times 10^5$ contents in the network. Popularity of contents follows Zipf distribution, and request arrival process is poisson with average rate of $12$/sec. We generate $3\times 10^6$ requests for cache warmup and another $6\times 10^6$ for measurement.

\textbf{Parameters. }We mainly consider the effect of: \emph{Zipf exponent $\alpha$}, \emph{cache to population ratio} and \emph{topology}. Zipf exponent $\alpha$ indicates the skewness of popularity distribution. When $\alpha$ increases, popular items become further popular and unpopular ones less. Cache to population ratio, introduced in \cite{hashrouting}, shows the proportion of total cache size to total content size. 

\subsection{Scenario 1: Single Cache}

Before evaluating the performance of UtilCache in a network, we first focus on the performance gap between UtilCache and the "intuitive thought" in Sec.~\ref{intuitive}. Let there be only one cache, and a miss request for $c$ at the cache leads to an average of $\bar{w}^c$ link cost. The "intuitive thought", i.e.\  to cache is top-$B$ contents with highest caching utility $\lambda^c w^c$, is the optimal solution (refered to \verb|Optimal|). Moreover, since LFU is the content replacement policy which can indeed cache the most popular contents, UtilCache is theoretically able to keep the most utilitarian contents when working with LFU. 

Fig.~\ref{single} shows the results. The gap between \verb|Optimal| and \verb|LFU| indicates the necessity of considering upstream link cost in addition to content popularity. UtilCache with LFU yields nearly the same caching gain as the optimal solution, but is more efficient. The gap between \verb|U-LFU| and \verb|U-LRU| is similar with that between \verb|LFU| and \verb|LRU|. Although the total caching gain attain by UtilCache with LRU is less, it has low computational complexity and can be executed efficiently.

\subsection{Scenario 2: Cache Network}
Now we run evaluations in a cache network, and compare UtilCache with LCE and GRD. Both LCE and UtilCache are implemented with LFU and LRU.

Caching gain under different circumstances are shown in Fig.~\ref{multi}. When the cache size increases, more contents are able to be cached to save link cost, for which caching gain increases correspondingly. When Zipf exponent $\alpha$ increases, it is more likely to request for the popular contents that are cached, for which the caching gain increases as well.

UtilCache performs much better when using same content replacement policy as LCE, which indicates we should consider the caching utility rather than merely popularity when reducing link cost. When working with LFU, UtilCache yields the best performance, which validates the efficacy of caching utility. With similar coordination overhead and efficiency, GRD performs worse than UtilCache (with LFU) because its estimation of caching utility has some inaccuracy.

%% file: 7appendix.tex
\appendix
\subsection{Transforming LCM Formulation}
\label{transform}
We define a \emph{Ground Set} as:
$$E = \{f^c_i\mid i\in\mathbf{N}, c\in\mathbf{C}\}$$. The integrality constraint \eqref{GM:3} enables that every cache decision $\mathbf{X}=\{x^c_i\}$ can be written as a set $A\subseteq E$, where
\begin{equation}
\label{equal}
x^c_i=1 \Leftrightarrow f^c_i\in A
\end{equation}

\begin{propo}
\label{cons}
Let $G_i=\{f^c_i\mid c\in\mathbf{C}\}$. The constraints of \eqref{U} is equivalent to $\mathcal{I}=\{A\subseteq E\mid |A\cap G_i|\leq B_i,\forall i\in\mathbf{N}\}$, while $\mathcal{M}=(E,\mathcal{I})$ is a partition matroid.
\end{propo}
\begin{proof}
Suppose $A$ and $\mathbf{X}$ are equivalent, which means \eqref{equal} is satisfied.
$|A\cap G_i|$ indicates how many elements in $\{f^c_i\mid c\in\mathbf{C}\}$ are in $A$ as well. Notice that here router $i$ is regarded as constant. Thus $|A\cap G_i|$ also refers to how many $x^c_i=1$ when $i$ is fixed, which is exactly $\sum_{c\in\mathbf{C}}x^c_i$. Therefore, we have $$\sum_{c\in\mathbf{C}}x^c_i\le B_i\Leftrightarrow |A\cap G_i|\leq B_i$$.
Once $\mathbf{X}$ is a feasible solution to \eqref{U}, the equivalent $A$ must satisfy: $|A\cap G_i|\leq B_i, \forall i\in\mathbf{N}$, for which constraints of \eqref{U} can be represented by a set of feasible solutions as well, i.e. $\mathcal{I}=\{A\subseteq E\mid |A\cap G_i|\leq B_i,\forall i\in\mathbf{N}\}$. 

\emph{Partition matroid} $(E',\mathcal{I}')$is a typical instance of matroids. In a partition matroid, the ground set $E'$ is partitioned into disjoint sets $E'_1$, $E'_2$,...,$E'_l$ and 
$$\mathcal{I}' = \{A\subseteq E'\mid |A\cap E'_i|\leq \beta_i, \forall i=1,...,l\}$$, for constant parameters $\beta_1, \beta_2,..., \beta_l$\cite{partionmatroid}. Obviously, $\{G_i\}$ is a partition of $E$, for which $\mathcal{M}=(E,\mathcal{I})$ is a partition matroid.
\end{proof}

\begin{propo}
\label{object}
The objective function of \eqref{U} can be written as a monotone submodular function.
\end{propo}
\begin{proof}
Although we model different problems, our objective function \eqref{U:1} is very similar with the objective function in \cite{femto2013}, which has been proven to be a monotone submodular function by Shanmugam et al.
\end{proof}

Now we can draw a conclusion from Proposition~\ref{cons} and \ref{object} that LCM problem formulation \eqref{U} can be written as maximizing a monotone submodular function with matroid constraints.

\subsection{Proof of Theorem~\ref{cu}}
\label{cachingu}
From the definition of $G(\mathbf{X})$ we know $x^c_i$ affects only the caching gain of $(c,p)$ where router $i$ is in path $p$. Let $\mathcal{R}_i^c=\{(c,p)\mid i\in p, (c,p)\in\mathcal{R}\}$ and $p_{i'}=i$ when $i\in p$. We define
\[
\begin{matrix}
G_i^c(\mathbf{X}) &=& \sum\limits_{(c,p)\in\mathcal{R}_i^c}\lambda_{(c,p)}\sum\limits_{k=1}^{K-1}w_{p_{k+1}p_k}(1-\prod\limits_{k'=1}^k (1-x^c_{p_{k'}}))\\
&=& \sum\limits_{(c,p)\in\mathcal{R}_i^c}\lambda_{(c,p)}[\sum\limits_{k=1}^{i'-1}w_{p_{k+1}p_k}(1-\prod\limits_{k'=1}^k (1-x^c_{p_{k'}}))+\\
&&\sum\limits_{k=i'}^{K-1}w_{p_{k+1}p_k}(1-\prod\limits_{k'=1}^{k} (1-x^c_{p_{k'}}))]
\end{matrix}
\]. Thus, we have
\[
\begin{matrix}
m^c_i(\mathbf{X}) &=& G(\mathbf{X}|x^c_i=1)-G(\mathbf{X}|x^c_i=0)\\
&=&G_i^c(\mathbf{X}|x^c_i=1)-G_i^c(\mathbf{X}|x^c_i=0)\\
&=&\sum\limits_{(c,p)\in\mathcal{R}_i^c}\lambda_{(c,p)}\sum\limits_{k=i'}^{K-1}w_{p_{k+1}p_k}\prod\limits_{k'=1}^{k}(1-x^c_{p_{k'}})\\
&=&\sum\limits_{(c,p)\in\mathcal{R}_i^c}\lambda_{(c,p)}\prod\limits_{k'=1}^{i'-1}(1-x^c_{p_{k'}})(w_{p_{i'+1}p_{i'}}\\
&&+\sum\limits_{k=i'+1}^{K-1}w_{p_{k+1}p_k}\prod\limits_{k'=i'+1}^{k}(1-x^c_{p_{k'}}))\\
\end{matrix}
\]

Let $\lambda'_{(c,p)}=\lambda_{(c,p)}\prod_{k'=1}^{i'-1}(1-x^c_{p_{k'}})$, and $w_{(c,p)}=w_{p_{i'+1}p_{i'}}+\sum_{k=i'+1}^{K-1}w_{p_{k+1}p_k}\prod_{k'=i'+1}^{k}(1-x^c_{p_{k'}})$. We can figure out that $\lambda'_{(c,p)}$ is the actual request arrival rate for content $c$ at router $i$ through path $p$, and $w_{(c,p)}$ represents the upstream link cost caused by a request for content $c$ through path $p$, and thus we have
\[
\begin{matrix}
m^c_i(\mathbf{X})&=&\sum\limits_{(c,p)\in\mathcal{R}_i^c}\lambda'_{(c,p)}w_{(c,p)}\\
&=&\frac{\sum\limits_{(c,p)\in\mathcal{R}_i^c}\lambda'_{(c,p)}}{\sum\limits_{(c,p)\in\mathcal{R}_i^c}\lambda'_{(c,p)}}\sum\limits_{(c,p)\in\mathcal{R}_i^c}\lambda'_{(c,p)}w_{(c,p)}\\
&=&\sum\limits_{(c,p)\in\mathcal{R}_i^c}\lambda'_{(c,p)}\cdot(\sum\limits_{(c,p)\in\mathcal{R}_i^c}\frac{\lambda'_{(c,p)}}{\sum\limits_{(c,p)\in\mathcal{R}_i^c}\lambda'_{(c,p)}}w_{(c,p)})
\end{matrix}
\]
Let $\lambda^c_i=\sum_{(c,p)\in\mathcal{R}_i^c}\lambda'_{(c,p)}$ and $\bar{w}^c_i=\sum_{(c,p)\in\mathcal{R}_i^c}\frac{\lambda'_{(c,p)}}{\sum_{(c,p)\in\mathcal{R}_i^c}\lambda'_{(c,p)}}w_{(c,p)}$. Thus, the former, $\lambda^c_i$, represents requests arrival rate for content $c$ at router $i$, while $\bar{w}^c_i$ represents the average per-request upstream link cost. In conclusion, we have $$m^c_i(\mathbf{X}) = \lambda^c_i\bar{w}^c_i$$